\documentclass[11pt,a4paper]{article}
\usepackage{epsfig}
\usepackage[T1]{fontenc}    
\usepackage{graphics}
\usepackage{graphicx}
\usepackage{pstricks,pst-coil,pst-fill,pst-plot}
\usepackage[fleqn]{amsmath}    
\usepackage{amssymb}    
\usepackage{amsfonts}   
\usepackage{verbatim}   
\usepackage{mathrsfs}   
\usepackage{dsfont}
\usepackage{euscript}
\usepackage{yfonts}
\usepackage{enumerate}     
\usepackage{amsthm}         
\usepackage{txfonts}
\usepackage{marvosym}
\usepackage{stmaryrd}
\usepackage{vmargin}        
\usepackage{wasysym}		

\setmarginsrb{1.8cm}{2cm}{1.8cm}{2cm}{1cm}{1cm}{1cm}{1.6cm}
 \makeatletter
 \@addtoreset{equation}{section}
 \makeatother

\providecommand{\MR}{\relax\ifhmode\unskip\space\fi MR }

\providecommand{\href}[2]{#2}

\long\def\symbolfootnote[#1]#2{\begingroup%
\def\thefootnote{\fnsymbol{footnote}}\footnote[#1]{#2}\endgroup}

\newtheorem{theorem}{Theorem}[section]
\newtheorem{prop}[theorem]{Proposition}
\newtheorem*{theorem*}{Theorem}

\newtheorem{lemma}[theorem]{Lemma}

\newtheorem{axiomsFF}[theorem]{Form Factor Axioms}

\newcommand\beq{\begin{equation}}
\newcommand\enq{\end{equation}}
\newcommand\bem{\begin{multline}}
\newcommand\enm{\end{multline}}

\def\beqa{\begin{eqnarray}}
\def\eeqa{\end{eqnarray}}
\def\ba{\begin{array}}
\def\ea{\end{array}}

\newcommand{\f}[2]{{\ensuremath{%
    \mathchoice%
    {\dfrac{#1}{#2}}
    {\dfrac{#1}{#2}}
    {\frac{#1}{#2}}
    {\frac{#1}{#2}}
}}}
\newcommand{\tf}[2]{\ensuremath{#1/#2}}

\def\be{\beta}
\def\ga{\gamma}
\def\Ga{\Gamma}

\def\de{\delta}

\def\eps{\epsilon}

\def\la{\lambda}

\def\sg{\sigma}

\def\ups{\upsilon}

\def\vp{\varphi}

\newcommand{\mc}[1]{\ensuremath{\mathcal{#1}}}
\newcommand{\mf}[1]{\ensuremath{\mathfrak{#1}}}
\newcommand{\msc}[1]{\ensuremath{\mathscr{#1}}}

\newcommand{\bs}[1]{\ensuremath{\boldsymbol{#1}}}

\DeclareFontFamily{OT1}{pzc}{}
\DeclareFontShape{OT1}{pzc}{m}{it}{<-> s * [1.10] pzcmi7t}{}
\DeclareMathAlphabet{\mathpzc}{OT1}{pzc}{m}{it}

\def \ii{ \mathrm i}

\newcommand{\ov}[1]{\ensuremath{\overline{#1}}}
\newcommand{\wt}[1]{\ensuremath{\widetilde{#1}}}
\newcommand{\wh}[1]{\ensuremath{\widehat{#1}}}

\newcommand{\Int}[2]{\ensuremath{\int\limits_{#1}^{#2}}}

\newcommand{\Fint}[2]{\ensuremath{\fint\limits_{#1}^{#2}}}

\newcommand{\sul}[2]{\ensuremath{\sum\limits_{#1}^{#2}}}
\newcommand{\pl}[2]{\ensuremath{\prod\limits_{#1}^{#2}}}

\newcommand{\R}{\ensuremath{\mathbb{R}}}
\newcommand{\Cx}{\ensuremath{\mathbb{C}}}

\newcommand{\Dp}[1]{\ensuremath{\partial_{#1}}}

\newcommand{\ex}[1]{\ensuremath{\e{e}^{#1}}}

\newcommand{\op}[1]{ \boldsymbol{ \texttt{#1} } }

\newcommand{\dd}{\mathrm{d}}
\newcommand{\e}[1]{\ensuremath{\mathrm{#1}}}

\newcommand{\intff}[2]{\ensuremath{ [  #1 \,; #2 ] }}

\newcommand{\intoo}[2]{\ensuremath{ ]  #1 \,; #2 [ }}

\newcommand{\intn}[2]{\ensuremath{[\![ \, #1 \,;\, #2 \,]\!]}}

\newcommand{\widesim}[2][1.5]{
  \mathrel{\underset{#2}{\scalebox{#1}[1]{$\sim$}}}
}

\begin{document}

\begin{center}
\begin{LARGE}
{\bf Bootstrap approach to  1+1 dimensional integrable quantum field theories: the case of the Sinh-Gordon model}
\end{LARGE}

\vspace{1cm}

\vspace{4mm}
{\large Karol K. Kozlowski \footnote{e-mail: karol.kozlowski@ens-lyon.fr}}%
\\[1ex]
Univ Lyon, ENS de Lyon, Univ Claude Bernard Lyon 1, CNRS, Laboratoire de Physique, F-69342 Lyon, France \\[2.5ex]

\par 

\end{center}

\vspace{40pt}


1+1 dimensional integrable quantum field theories correspond to a sparse subset of quantum field theories where the calculation of 
physically interesting observables can be brought to explicit, closed and manageable expressions thanks to the factorisability 
of the S matrices which govern the scattering in these models. In particular,  the correlation functions  
are expressed in terms of explicit series of multiple integrals, this non-perturbatively for all values of the 
coupling. However, the question of convergence of these series, and thus the mathematical well-definiteness of these correlators, 
is mostly open. This paper reviews the overall setting used to formulate such models and discusses  
the recent progress relative to solving the convergence issues in the case of the 1+1 dimensional massive integrable Sinh-Gordon quantum field theory.

\vspace{40pt}

\tableofcontents

\section{Introduction}

\subsection{Scattering matrices for quantum integrable field theories}

It was discovered in the early XX century that the description of matter at low-scales demands to wave-off some of the existing at the time 
paradigms governing the motion and very structure of particles in interactions. This led to the development of the theory of relativity on the one hand, and 
quantum mechanics on the other. In the latter setting, the state of a physical system is described by a vector, the wave function, belonging to some Hilbert space 
and supposed to encapsulate all the physical degrees of freedom of that system. 
On the classical level, the time evolution of 
particles' momenta and positions is governed by a set of generically non-linear ordinary differential equations which can be written in the form of Hamilton's equations. 
In its turn, the time evolution of a wave function is governed by a first order ordinary linear differential equation driven by the Hamiltonian operator.  
This operator is obtained through a quantisation procedure: its symbol is given by the classical Hamiltonian of the system 
or, said differently, it is obtained from the classical Hamiltonian upon replacing the classical momenta and positions by operators.
While the success of the approach was astonishing relatively to the amount of experiments which could
have been explained, soon after the early development of the theory it became clear that in order to describe physics at even smaller scales 
or higher energies, one needs to develop a quantum theory of fields which would bring together the quantum and relativistic features in the setting of uncountably many degrees of freedom. 
 In loose words, such a theory would be reached by producing operator valued generalised functions, \textit{viz}. formal kernels of
distributions, depending on the space-time coordinates which would satisfy analogues of non-linear, relativistically invariant, evolution equations arising in classical field theory. 
While it was rather straightforward to construct the quantum theory of the free field (and nowadays such a construction is fully rigorous)
the construction of interacting theories, the sole relevant for physics, appeared to be a tremendously hard task, this even on a formal level of rigour. 
The various approaches that were developed quickly met serious problems, the most prominent being the divergence of coefficients supposed to describe the formal perturbative expansions  of physical observables
around the free theories.  Eventually, these problems could have been formally circumvented in certain cases by the use of the so-called renormalisation procedure. 
The latter, while being able to produce numbers which were measured with great agreement in collider experiments, eluded for very long any
attempts at making it rigorous. Some progress was eventually achieved for several instances of truly interacting, \textit{viz}. non-free, quantum field theories 
within the so-called constructive quantum field theory approach, see \cite{SummersStateOfTheArtConstructiveQFTAsOf2012} for a review. While successful in rigorously showing 
the existence and certain overall properties of such theories, the approach did not lead yet to rigorous and manageable expressions for the correlation functions, 
which are the quantities measured in experiments and thus of prime interest to the theory. 

Among the various alternatives to renormalisation, one may single out the $\op{S}$-matrix program which aimed at describing a quantum field theory  
directly in terms of the quantities that are measured in experiments. This led to a formulation  of the theory in terms of matrix valued functions in $n$ complex variables, with $n=0,1,2,\dots$, 
that correspond to the entries of the $\op{S}$-matrix between asymptotic states. 
The  $\op{S}$-matrix program was actively investigated in the 60s and 70s and numerous attempts were made to characterise the $\op{S}$-matrix which is the central object in this approach, see \textit{e.g.} \cite{EdenLandshoffOlivePolkinghorneAnalyticPropertiesofSMatrix,IagolnitzerSmatrix}. However, these investigations led to rather 
unsatisfactory results in spacial dimensions higher than one, mainly due to the incapacity of constructing viable, explicit, $\op{S}$-matrices
for non-trivial models. 

The interest in the $\op{S}$-matrix approach was revived by the pioneering work of Gryanik and Vergeles \cite{GryanikVergelesSMatrixAndOtherStuffForSinhGordon}. 
These authors set forth the 
first features of an integrable structure based-method for determining $\op{S}$-matrices for the 1+1 dimensional quantum field theories whose classical analogues 
exhibit an infinite set of independent local integrals of motion.  Indeed, the existence of analogous conservation laws on the quantum level 
heavily constrains the possible form of the scattering
basically by reducing it to a concatenation of two-body processes and hence making the calculations of $\op{S}$-matrices feasible. 
The work \cite{GryanikVergelesSMatrixAndOtherStuffForSinhGordon}
focused on the case of models only exhibiting one type of asymptotic particles, the main example being given by the quantum Sinh-Gordon model. 
This 1+1 dimensional quantum field theory which will be taken as a guiding example from now on. It corresponds to the appropriate quantisation of the  classical evolution equation of a scalar field $\varphi(x,t)$ under the partial differential equation 
\beq
\Big(\Dp{t}^2- \Dp{x}^2\Big)  \vp \, + \, \f{m^2}{g} \sinh(g\varphi) \, = \, 0   \qquad (x,t) \in \R^2\;. 
\enq
For this model, the asymptotic "in" states of the theory are described by vectors $\bs{f}=(f^{(0)},\dots, f^{(n)},\dots)$ which belong to the Fock Hilbert space 
\beq
\mf{h}_{\e{in}} \, = \, \bigoplus\limits_{n=0}^{+\infty} L^2(\R^n_{>}) \qquad \e{with} \qquad \R^n_{>} \; = \; \Big\{ \bs{\be}_{n}=(\be_1,\dots, \be_n) \in \R^n \; : \; \be_1>\dots>\be_n  \Big\} \;. 
\label{definition d eh in}
\enq
This means that $f^{(n)}\in L^2( \R^n_{>} )$ has the physical interpretation of an incoming $n$-particle wave-packet density in rapidity space. More precisely, on physical grounds, one interprets elements of the
Hilbert space $\mf{h}_{\e{in}}$ as parameterised by $n$-particles states, $n\in \mathbb{N}$,  arriving, in the remote past, with well-ordered rapidities $\be_1>\dots>\be_n$
prior to any scattering which would be enforced by the interacting nature of the model.

For the 1+1 dimension quantum Sinh-Gordon model the $\op{S}$-matrix proposed in \cite{GryanikVergelesSMatrixAndOtherStuffForSinhGordon} is purely diagonal and 
thus fully described by one scalar function of the relative "in" rapidities of the two particles:
\beq
\op{S}(\beta)\, = \, \f{ \tanh\big[ \tfrac{1}{2}\beta - \ii \pi  \mf{b}   \big]  }{ \tanh\big[ \tfrac{1}{2}\beta + \ii \pi  \mf{b}   \big]   }  \qquad \e{with} \qquad  \mf{b}\, = \,   \f{1}{2} \f{ g^2  }{ 8\pi + g^{2}  } \;. 
\label{definition matrice S}
\enq
This $\op{S}$-matrix satisfies the unitarity $\op{S}(\be)\op{S}(-\be)=1$, and crossing  $\op{S}(\be)=\op{S}(\ii \pi-\be)$ symmetries. 
These are, in fact,  fundamental symmetry features of an $\op{S}$-matrix and arise in many other integrable quantum field theories. 
 Within the physical picture, throughout the flow of time, the "in" particles approach each other, interact, scatter and finally travel again as asymptotically free outgoing, \textit{viz}. "out", 
particles.
Within such a scheme, an "out" $n$-particle state is then paramaterised by $n$ well-ordered rapidities  $\be_1<\dots < \be_n$ and  can be seen as 
a component of a vector belonging to the Hilbert space 
\beq
 \mf{h}_{\e{out}} \, = \, \bigoplus\limits_{n=0}^{+\infty} L^2(\R^n_{<}) \qquad \e{with} \qquad \R^n_{<} \; = \; \Big\{ \bs{\be}_{n}=(\be_1,\dots, \be_n) \in \R^n \; : \; \be_1<\dots<\be_n  \Big\} \;.  
\enq
The $\op{S}$-matrix will allow one to express the "out" state $\bs{g}=(g^{(0)},\dots, g^{(n)},\dots)$ which results from the scattering of an "in" state 
$\bs{f}=(f^{(0)},\dots, f^{(n)},\dots)$ as 
\beq
g^{(n)}(\be_1,\dots,\be_n) \, = \,   \pl{a<b}{n} S(\be_a-\be_b) \cdot  f^{(n)}(\be_n,\dots,\be_1)  \;. 
\enq
Note that in this integrable setting, there is \textit{no} particle production and that the scattering is a concatenation of two-body processes.

 
Over the years, it turned out to be possible to characterise thoroughly the $\op{S}$-matrices for more involved quantum field theories underlying to other 
integrable classical field theories in 1+1 dimensions. Such models possess several types of asymptotic particles which can also form bound states. 
Then, the "in" Fock Hilbert space is more complicated and takes the form $  \bigoplus\limits_{n=0}^{+\infty} L^2(\R^n_{>},\otimes^n \Cx^{p})$ where 
the $L^2$-space refers to $\otimes^n \Cx^{p}$ valued functions on $\R^n_{>}$, with $p$ corresponding to the number of different asymptotic particles in the given theory. 
The most celebrated example corresponds to the Sine-Gordon quantum field theory. Building on Faddeev-Korepin's \cite{KorepinFaddeevQuantisationOfSolitions} semi-classical quantisation results of the solitons 
in the classical Sine-Gordon model, one concludes that the underlying quantum field theory possesses two distinct types of asymptotic 
particles of equal mass, the soliton and the anti-soliton, as well as a certain number, which depends on the coupling constant, of bound states thereof. 
These all have distinct masses and are called breathers. Zamolodchikov argued the explicit form of the $\op{S}$-matrix governing the soliton-antisoliton scattering \cite{ZalmolodchikovSMatrixSolitonAntiSolitonSineGordon}
upon using  the factorisability of the $n$-particle $\op{S}$-matrix into two-particle processes, 
the independence of the order in which a three particle scattering process arises from a concatenation of two-particle processes as well as
the fact that equal mass particles may \textit{solely} exchange their momenta during scattering, this due to the existence of many conservation laws. 
This enforces that the $\op{S}$ matrix 
satisfies the Yang-Baxter equation, which originally appeared in rather different contexts
\cite{BaxterPartitionfunction8Vertex-FreeEnergy,YangFactorizingDiffusionWithPermutations}, and strongly restricts its form. 
We do stress that the Yang-Baxter equation is the actual cornerstone of quantum integrability, so that it is not astonishing to recover it also in this setting. 
The missing pieces of the Sine-Gordon $\op{S}$-matrix capturing the soliton-breather and breather-breather scarrering were then proposed in \cite{KarowskiThunCompleteSMatrixThirring}. 
Nowadays, $\op{S}$-matrices of many other models have been proposed, see \textit{e.g.} \cite{ArinshteinFateyevZamolodchikovSMatrixTodaChain,ZalZalBrosFactorizedSMatricesIn(1+1)QFT}.


\subsection{The operator content \& the Bootstrap program}

\subsubsection{The basic operators}
\label{SousSoussectionOperateursDeBase}

Having in mind the \textit{per se} full construction of the quantum field theory, identifying the content in asymptotic particles, \textit{viz}. the "in" particles' Hilbert space $\mf{h}_{\e{in}}$, 
and the $\op{S}$-matrix which describes their scattering only arises as the first step. Indeed, one should build, in a way that is compatible with the form of the scattering 
encapsulated in the $\op{S}$-matrix of interest, a family $\op{O}_{\alpha}$ of operator-valued distributions, $\alpha$ running through some set $\mc{S}$. More precisely, 
the $\op{O}_{\alpha}$ should be distributions acting on smooth, compactly supported functions $d(\bs{x})$ of 
the Minkowskian space-time coordinate  
\beq
\bs{x}\, =\, \big( x_0 , x_1\big) \in \R^{1,1} \qquad \e{with} \qquad \bs{x} \cdot \bs{y} = x_0y_0-x_1y_1 \;. 
\enq
Then $\op{O}_{\alpha}[d]$ is some densely defined operator on $\mf{h}_{\e{in}}$ whose domain could, in principle, depend on $d$. It is useful from the point 
of view of connecting this picture to physics to express $\op{O}_{\alpha}$ directly in terms of its generalised operator valued function
\beq
\op{O}_{\alpha}[d] \, = \, \Int{ \R^2 }{} \dd^2  \bs{x} \, d(\bs{x})\op{O}_{\alpha}(\bs{x}) \;. 
\enq
In fact, in physics' terminology, it is the $\op{O}_{\alpha}(\bs{x})$s which correspond to the quantum fields of the theory. 
Moreover, as will be apparent in the following, it turns out that in most handlings $\op{O}_{\alpha}(\bs{x})$
does actually make sense as a \textit{bona fide} operator valued \textit{function} on the Minkowski space having a well-defined dense domain. 
Hence, unless it is mandatory so as to make an appropriate sense out of the formula, we will make use of the generalised function 
notation $\op{O}_{\alpha}(\bs{x})$. 

On top of being compatible with the scattering date, the operators  $\op{O}_{\alpha}(\bs{x})$ should  form an algebra, \textit{viz}. the product $\op{O}_{\alpha}(\bs{x}) \op{O}_{\alpha^{\prime}}(\bs{y})$
should be a well-defined dense operator for almost all $\bs{x}$ and $\bs{y}$, and satisfy causality, \textit{viz}. that for purely Bosonic theories
as the Sinh-Gordon model
\beq
\big[ \op{O}_{\alpha}(\bs{x}) \, , \,  \op{O}_{\alpha^{\prime}}(\bs{y}) \big] \; \equiv  \; 
\op{O}_{\alpha}(\bs{x})  \op{O}_{\alpha^{\prime}}(\bs{y}) \, - \,  \op{O}_{\alpha^{\prime}}(\bs{y}) \op{O}_{\alpha}(\bs{x})  \; = \; 0 \quad \e{if} \quad 
(\bs{x}-\bs{y})^2 \, < \, 0 \;, 
\enq
namely when $\bs{x}-\bs{y}$ is space-like. The family $\op{O}_{\alpha}(\bs{x})$ should in particular contain the \textit{per se} quantised counterparts of the classical fields 
arising in the original evolution equation, for instance $\bs{\Phi}(\bs{x})$ or $\ex{\ga \bs{\Phi} }(\bs{x})$ in the Sinh-Gordon quantum field theory case. 
Moreover, these operators should comply with the various
other symmetries imposed on a quantum field theory, such as invariance under Lorentz boosts of space-time coordinates or translational invariance. 
In the quantum Sinh-Gordon field theory on which we shall focus from now on, the latter means that the model is naturally endowed with a unitary operator $\op{U}_{\op{T}_{\bs{y}}}$ such that 
for any operator $\op{O}(\bs{x})$
\beq
\op{U}_{\op{T}_{\bs{y}}} \cdot \op{O}(\bs{x}) \cdot \op{U}_{\op{T}_{\bs{y}}}^{-1} \, = \, \op{O}(\bs{x}+\bs{y}) \;. 
\label{ecriture action adjointe operateur de translation}
\enq
The operator $\op{U}_{\op{T}_{\bs{y}}}$ acts diagonally on $\mf{h}_{\e{in}}$ given in \eqref{definition d eh in}:
\beq
\op{U}_{\op{T}_{\bs{y}}} \cdot \bs{f} \; = \; \Big( \op{U}_{\op{T}_{\bs{y}}}^{(0)}\cdot f^{(0)},\dots, \op{U}_{\op{T}_{\bs{y}}}^{(n)}\cdot f^{(n)},\dots \Big)
\quad \e{with} \quad \bs{f}=(f^{(0)},\dots, f^{(n)},\dots)
\enq
and where 
\beq
\op{U}_{\op{T}_{\bs{y}}}^{(n)} \cdot f^{(n)}(\bs{\be}_{n}) \; = \; \exp\bigg\{\ii \sul{a=1}{n}\bs{p}(\be_a)\cdot \bs{y} \bigg\}  f^{(n)}(\bs{\be}_{n})
\label{ecriture action operateur de translation}
\enq
with $\bs{p}(\be) = \big( m\cosh(\be), m \sinh(\be) \big)$ and $\bs{\be}_n \, = \, (\be_1,\dots, \be_n)$.

 Should the construction of quantum fields fulfilling to the above be achieved, the ultimate goal would consist in computing in closed and explicit form the model's vacuum-to-vacuum $n$-point correlation functions:
\beq
\Big<  \op{O}_{\alpha_1}(\bs{x}_1) \cdots \op{O}_{\alpha_n}(\bs{x}_n)  \Big> \, = \, \e{Tr}_{ \mf{h}_{\e{in}} } \Big[ \op{P}_0 \op{O}_{\alpha_1}(\bs{x}_1) \cdots \op{O}_{\alpha_n}(\bs{x}_n) \op{P}_0\Big]
\enq
with $\op{P}_0$ being the orthogonal projection on the $0$-particle Fock space. We do stress that the above objects are still generalised functions and, as such, should be 
considered in an appropriate distributional interpretation. That will be made precise below. 
 
 \subsubsection{The bootstrap program for the zero particle sector}

By virtue of the above, in the case of the $\mf{h}_{\e{in}}$ Hilbert space, one may represent an operator $\op{O}(\bs{x})$ as an integral operator 
acting on the $L^2$-based Fock space
\beq
\op{O}(\bs{x}) \cdot \bs{f} \; = \; \Big( \, \op{O}^{(0)}(\bs{x})  \cdot \bs{f} , \cdots ,  \op{O}^{(n)}(\bs{x})\cdot \bs{f} , \cdots \Big) 
\enq
with $ \op{O}^{(n)}(\bs{x}) \, : \,  \mf{h}_{\e{in}} \rightarrow L^2(\R^{n}_{>})$. Later on, we will discuss more precisely the structure 
of the operators $ \op{O}^{(n)}(\bs{x}) $ that one needs to impose so as to end up with a consistent quantum field theory. However, first, we focus our attention on the 
$0^{\e{th}}$ space operators whose action may be represented, whenever it makes sense, as 
\beq
\op{O}^{(0)}(\bs{x})\cdot \bs{f} \; = \; \sul{m \geq 0}{} \; \Int{  \R^m_{>} }{} \hspace{-1mm}  \dd^m \be  \;  
\mc{M}_{0;m}^{(\op{O})}(\bs{\be}_m) \pl{a=1}{m} \Big\{ \ex{- \ii \bs{p}(\be_a)\cdot \bs{x} } \Big\}
f^{(m)}\big( \bs{\be}_m \big)   \;. 
\label{ecriture chp quantique comme op integral secteur 0}
\enq
The oscillatory $\bs{x}$-dependence is a simple consequence of the translation relation \eqref{ecriture action adjointe operateur de translation} along with the explicit form of the action of the translation operator \eqref{ecriture action operateur de translation}.  

In order for $\op{O}^{(0)}(\bs{x})$ to comply with the scattering data encoded by $\op{S}$, one needs to impose a certain amount of constraints on the integral kernels $\mc{M}_{0;m}^{(\op{O})}(\bs{\be}_m)$. 
First of all, general principles of quantum field theory impose that, in order for these to correspond to kernels of quantum fields,the  $\mc{M}_{0;m}^{(\op{O})}(\bs{\be}_m)$
have to correspond to 
a  $+$ boundary value $\mc{F}_{m;+}^{(\op{O})}(\bs{\be}_m)$ on $\R^m_{>}$ of a meromorphic function $\mc{F}_{m}^{(\op{O})}(\bs{\be}_m)$
of the variables $\be_a$ belonging to the strip 
\beq
\msc{S} \, = \,   \big\{ z \in \Cx \; : \;  0 < \Im (z) < 2 \pi \big\} .
\enq
Traditionally, in the physics literature,  the functions $\mc{F}_{m}^{(\op{O})}(\bs{\be}_m)$
are called form factors. 

Further, one imposes a set of equations on the $\mc{F}_{m}^{(\op{O})}$s. These constitute the so-called form factor bootstrap program. On mathematical grounds, one should understand the form factor bootstrap program
as a set of \textit{axioms} that one imposes as a starting point of the theory given the data $\big( \mf{h}_{\e{in}}, \op{S} \big)$. Upon solving them, one has to check 
\textit{a posteriori} that their solutions do provide one, through \eqref{ecriture chp quantique comme op integral secteur 0} and \eqref{ecriture chp quantique comme op integral secteur general}, 
with a collection of operators satisfying all of the requirements of the theory discussed earlier on. 

The bootstrap program axioms take the form of a Riemann-Hilbert problem for a collection of functions in 
many variables.  In the case of the Sinh-Gordon model, since there are no bound states, these take the below form. 

\begin{axiomsFF}
 
Find functions $\mc{F}_{n}^{(\op{O})}$,  $n \in \mathbb{N}$, such that, for each $k \in \intn{1}{n}$ and fixed $\be_a \in \msc{S}$, $a \not=k$, 
the maps $\be_k \mapsto \mc{F}_{n}^{(\op{O})}(\bs{\be}_n)$ are 
\begin{itemize}
 
 \item meromorphic on $\msc{S}$;
 
 \item admit $+$, resp. $-$, boundary values $\mc{F}_{n;+}^{(\op{O})}$ on $\R$, resp. $\mc{F}_{n;-}^{(\op{O})}$ on  $\R+2\ii\pi$;
 
 \item are bounded at infinity by $ C\cdot \cosh\big( \ell \Re(\be_k) \big)$ for some $n$ and $k$ independent  $\ell$.
 
\end{itemize}
The $\mc{F}_{n}^{(\op{O})}$ satisfy the multi-variable system of Riemann-Hilbert problems:
\begin{itemize}
\item[i)]  $\mc{F}_{n}^{(\op{O})}(\be_1,\dots, \be_a, \be_{a+1},\dots,  \be_n) \; = \; \op{S}(\be_{a}-\be_{a+1}) \cdot \mc{F}_{n}^{(\op{O})}(\be_1,\dots, \be_{a+1}, \be_{a},\dots,  \be_n)$;
\item[ii)]  For $\be_1 \in \R$, and given generic $\bs{\be}_n^{\prime}=(\be_2,\dots, \be_n) \in \msc{S}^{n-1}$, \newline  
$ \mc{F}_{n;-}^{(\op{O})}(\be_1+2\ii\pi, \bs{\be}_n^{\prime}) \; = \; \mc{F}_{n;+}^{(\op{O})}(\bs{\be}_n^{\prime},\be_1)
 \, = \, \pl{a=2}{n} \big\{ \op{S}(\be_{a}-\be_{1}) \big\}\cdot \mc{F}_{n;+}^{(\op{O})}( \bs{\be}_n)$; 
\item[iii)] The only poles of  $\mc{F}_{n}^{(\op{O})}$ are simple,  located at 
$\ii\pi$ shifted rapidities and
\beq
-\ii \e{Res}\Big(\mc{F}_{n+2}^{(\op{O})}(\alpha+\ii\pi, \be, \bs{\be}_n) \cdot \dd \alpha \, , \, \alpha=\be  \Big) \; = \; \Big\{ 1\, - \, \pl{a=1}{n}  \op{S}(\be-\be_{a}) \Big\} \cdot 
\mc{F}_{n}^{(\op{O})}( \bs{\be}_n) \;. 
\nonumber
\enq
\item[iv)] $ \mc{F}_{n}^{(\op{O})}( \bs{\be}_n + \theta \ov{\bs{e}}_n )  \; = \; \ex{\theta \op{s}_{\op{O}} } \cdot \mc{F}_{n}^{(\op{O})}(\bs{\be}_n)$ for some number $\op{s}_{\op{O}}$ and with 
 $\ov{\bs{e}}_n=(1,\dots,1)$   
\end{itemize}

\end{axiomsFF}
 
Note that the reduction occurring at the residues of $ \mc{F}_{n}^{(\op{O})}( \bs{\be}_n)$ at the points $\be_{ab}=\ii \pi$ where $\be_{ab}=\be_a-\be_b$ can be readily inferred from $\mathrm{i)}$ and $\mathrm{iii)}$.

One may already comment on the origin of the axioms. The first one illustrates the scattering properties of the model on the level of the operator's kernel. 
The second and third axioms may be interpreted heuristically as a consequence of the LSZ reduction \cite{LehmannSymanzikZimmermanLSZReductionFormulaOriginalPaper}, and locality of the operator,
see \textit{e.g.} \cite{BabujianFringKarowskiZapletalExactFFSineGordonBootsstrapI,SmirnovFormFactors} for heuristics on that matter. Finally, 
the last axiom is a manifestation of the Lorentz invariance of the theory. The number $\op{s}_{\op{O}}$ arising in $\mathrm{iv)}$ is called the spin of the operator.
Moreover, the number $\ell$ depends on the type of operator being considered.  
Finally, for more complex models, one would also need to add an additional axiom which would encapsulate the way how the presence of  bound states
in the model governs certain additional poles in the form factors, \textit{c.f.} \cite{SmirnovFormFactors}.

\subsubsection{The bootstrap program for the multi-particle sector}

It is convenient to represent the action of the operators $\op{O}^{(n)}(\bs{x})$ in the form 
\beq
\Big(\op{O}^{(n)}(\bs{x})\cdot \bs{f} \Big)(\bs{\ga}_n) \; = \; \sul{m \geq 0}{} \; \pl{a=1}{n} \Big\{ \ex{ \ii \bs{p}(\ga_a)\cdot \bs{x} } \Big\}
\cdot   \op{M}_{\op{O}}^{(m)}\big( \bs{x} \, \mid     \bs{\ga}_n \big) \big[f^{(m)}\big] \;. 
\label{ecriture chp quantique comme op integral secteur general}
\enq
There $\op{M}_{\op{O}}^{(m)}\big( \bs{x} \, \mid     \bs{\ga}_n \big)$ are distribution valued functions which act on appropriate spaces of sufficiently regular functions in $m$ variables. 
The regularity assumptions will clear out later on, once that we provide the explicit expressions \eqref{ecriture forme explicite de la distribution} for these distributions. In fact, it is convenient, in order to avoid heavy notations, to represent their action as 
generalised integral operators
\beq
  \op{M}_{\op{O}}^{(m)}\big( \bs{x} \, \mid     \bs{\ga}_n \big) \big[f^{(m)}\big]  \; =   \Int{ \R^m_{>} }{} \hspace{-1mm}  \dd^m \be  \;   \mc{M}^{(\op{O})}_{n;m}\big(     \bs{\ga}_{n}  ; \bs{\be}_{m} \big) 
\pl{a=1}{m} \Big\{ \ex{- \ii \bs{p}(\be_a)\cdot \bs{x} } \Big\} 
f^{(m)}\big( \bs{\be}_m \big) \;, 
\enq
in which one understands of the kernels $\mc{M}^{(\op{O})}_{n;m}\big( \bs{\ga}_{n}  ; \bs{\be}_{m} \big) $ as generalised functions.  
 
  The last axiom of the bootstrap program  provides one  with a way to compute these kernels. Heuristically, it can be seen as a consequence of the LSZ reduction \cite{LehmannSymanzikZimmermanLSZReductionFormulaOriginalPaper}. 
\begin{itemize} 
 \item[v)] $\mc{M}^{(\op{O})}_{n;m}\big(  \bs{\alpha}_n ; \bs{\be}_{m} \big) \; = \;  \mc{M}^{(\op{O})}_{n-1;m+1}\big(  \bs{\alpha}_n^{\prime} ; (\alpha_1+\ii \pi, \bs{\be}_m) \big) \\
\; + \; 2\pi \sul{a=1}{m}  \de_{\alpha_1;\be_a} \pl{k=1}{a-1} \op{S}(\be_k-\alpha_1) \cdot 
\mc{M}^{(\op{O})}_{n-1;m-1}\big(  \bs{\alpha}_n^{\prime}; (\be_1, \dots, \wh{\be}_a , \dots,  \be_m) \big)$. 
\end{itemize}

In the above expression, $ \wh{\be}_a$ means that the variable $\be_a$ should be omitted and $\de_{x;y}$ refers to the Dirac mass distribution centred at $x$ and acting on functions of $y$. 
Finally, the evaluation at $\alpha_1+\ii\pi$ is understood in the sense of a boundary value of the meromorphic continuation in the strip $0\leq \Im(z) \leq \pi$ from $\R$ up to $\R+\ii\pi$. 
This axiom is to be complemented with the initialisation
condition $\mc{M}^{(\op{O})}_{0;n}\big( \emptyset ; \bs{\be}_{n} \big)\;= \;   \mc{F}_{n;+}^{(\op{O})}(\bs{\be}_n)$ when $\bs{\be}_n \in \R^{n}_{>}$. 
It is direct to establish that the recursion may be solved in closed form allowing one to determine the distributional kernel $\mc{M}^{(\op{O})}_{n;m}\big(  \bs{\alpha}_n ; \bs{\be}_{m} \big)$
in terms of $\mc{F}_{n}^{(\op{O})}(\bs{\be}_n)$.
\begin{multline}
\mc{M}^{(\op{O})}_{n;m}\big(  \bs{\alpha}_n ; \bs{\be}_{m} \big) \; = \; \sul{p=0}{ \e{min}(n,m) } \sul{  \substack{ k_1<\dots < k_p \\ 1 \leq k_a \leq n}  }{} 
\sul{  \substack{ i_1\not=\dots\not= i_p \\ 1 \leq i_a \leq m}  }{}  
\pl{a=1}{p}\Big\{2\pi \de_{\alpha_{k_a} ; \be_{i_a}}  \Big\}  \op{S}\big(\overleftarrow{\bs{\alpha}}_{n}  \mid \overleftarrow{\bs{\alpha}}_{n}^{(1)} \big)  \\
\times \op{S}\big(\bs{\be}_{n}^{(1)}  \mid \bs{\be}_{n} \big)    \cdot \mc{F}_{n+m-2p;- }\big(  \overleftarrow{\bs{\alpha}}_{n}^{(2)} +\ii\pi \ov{\bs{e}}_{n-p}, \bs{\be}_m^{(2)} \big) 
\label{ecriture forme explicite de la distribution}
\end{multline}
There, we have used the shorthand notations $\bs{\alpha}_{n}^{(1)}=(\alpha_{k_1},\dots, \alpha_{k_p})$ and $\bs{\alpha}_{n}^{(2)}=(\alpha_{\ell_1},\dots, \alpha_{\ell_{n-p}})$
where $\{\ell_1,\dots, \ell_{n-p} \} \, = \, \intn{1}{n} \setminus \{ k_a\}_1^p$, $\ell_1<\dots < \ell_{n-p}$, and analogously
$\bs{\be}_{m}^{(1)}=(\be_{i_1},\dots, \be_{i_p})$ and $\bs{\be}_{m}^{(2)}=(\be_{j_1},\dots, \be_{j_{m-p}})$
where $\{j_1,\dots, j_{m-p} \} \, = \, \intn{1}{m} \setminus \{ i_a\}_1^p$, $j_1<\dots < j_{m-p}$. Moreover, we have introduced 
\beq
 \op{S}\big(\overleftarrow{\bs{\alpha}}_{n}  \mid \overleftarrow{\bs{\alpha}}_{n}^{(1)} \big)  =   \pl{a=1}{p} \pl{ \substack{  b =1 \\  k_a> \ell_b}  }{ n-p } \op{S}\big( \alpha_{k_a} - \alpha_{\ell_b} \big)
\; , \quad 
\op{S}\big(\bs{\be}_{n}^{(1)}  \mid \bs{\be}_{n} \big)    =  \pl{a=1}{p}  \pl{ \substack{  b =1 \\ b<i_a}  }{ m }  S(\be_a-\be_{i_a}) \cdot \pl{ \substack{ a>b \\  i_a>i_b} }{} S\big( \be_{i_a} - \be_{i_b} \big)
\nonumber
\enq
Finally, we agree upon  $\overleftarrow{\bs{\ga}}_N = (\ga_N,\dots, \ga_1)$ for any $\bs{\ga}_N=(\ga_1,\dots, \ga_{N})$.  

It is clear on the level of the explicit expression \eqref{ecriture forme explicite de la distribution} that this generalised function is well defined, 
even though it involves a multiplication of distributions.

 \subsubsection{The road towards the bootstrap program}
 
 The first calculation of certain of the operators' kernels was initiated by  Weisz \cite{WeisztwoparticleFFForSolitionAntiSolitionInSineGordon}
 who built on the full characterisation of the $\op{S}$ matrix of the Sine-Gordon model to argue with the help of general principles of quantum field theory an expression for
 the kernel $\mc{M}^{(\op{O})}_{1;1}(\alpha;\be)$ of the electromagnetic current operator only involving one dimensional variables $\alpha, \be$. The setting up of a systematic approach allowing one to calculate
all the collection of kernels characterising an operator starting from a given model's $\op{S}$-matrix has been initiated by Karowski and Weisz \cite{KarowskiWeiszFormFactorsFromSymetryAndSMatrices}
who proposed  a set of equation satisfied by that model's equivalent of $\mc{F}_{n}^{(O)}( \bs{\be}_n)$. These allowed them to provide closed expressions for two particle form factors in several models. 
However these equations were still far from forming the full bootstrap program as described above.

After long investigations \cite{SmirnovUseGLMEqnsTocomputeSineGordonFF,SmirnovGLMEqnsDerivationForSineGordon,SmirnovUseGLMEqnsForSineGordonFF}  
which revealed a deeper structure of the form factors of the Sine-Gordon model, Smirnov \cite{SmirnovIntegralRepSolitonFFSineGordonBootstrap} formulated the equivalent of axioms $\mathrm{i)-ii)}$ in that model. 
Subsequently, Kirillov and Smirnov \cite{KirillovSmirnovFirstCompleteSetBootstrapAxiomsForQIFT} proposed the full set of the bootstrap program axioms, exemplified
in the case of the Massive Thirring model. See also \cite{SmirnovFormFactors}

\section{Solving the bootstrap program}



The resolution of the bootstrap program was systematised over the years and these efforts led to explicit expressions for the form factors of local operators in
numerous 1+1 dimensional massive quantum field theories, see \textit{e.g.} 
\cite{SmirnovFormFactors}. The first expressions for the form factors were rather combinatorial in nature. Later, a substantial progress was achieved in simplifying the latter, 
in particular by exhibiting a deeper structure at their root. Notably, one can mention the free field based approach,  also called angular quantisation,
to the calculation of form factors. It was introduced by Lukyanov \cite{LukyanovFirstIntroFreeField} and allowed to obtain convenient representations for certain form factors solving the bootstrap program. 
In particular, the construction lead to closed and manageable expressions \cite{LukyanovConjectureFFExponentialFieldSineGSolASolAndBReather,BrazhnikovLukyanovFreeFieldRepMassiveFFIntegrable} for the form factors of the exponential of the field operators in the Sinh-Gordon and the 
Bullough-Dodd models. Later, Babujian, Fring, Karowski, Zapetal \cite{BabujianFringKarowskiZapletalExactFFSineGordonBootsstrapI} and Babujian, Karowski \cite{BabujianKarowskiExactFFSineGordonBootsstrapII,BabujianKarowskiBreatherFFSineGordon}
developed the more powerful $\mc{K}$-transform approach which will be described below  on the example of the Sinh-Gordon model.
The construction of \cite{BabujianKarowskiExactFFSineGordonBootsstrapII,BabujianKarowskiBreatherFFSineGordon} was improved in   \cite{FeiginLashkevichFreeFieldApproachToDescendents,LashkevichPugaiDescendantFFResonanceIdentites}
so as to encompass  more complicated operators, the so-called descendants of the Sinh-Gordon exponential of the field operator.

\subsection{The 2-particle sector solution}

The constructions of solutions to the Bootstrap program starts from obtaining a specific solution to the equations $\mathrm{i)-iv)}$ when $n=2$, \textit{i.e.} for two variables. This was first achieved in 
\cite{KarowskiWeiszFormFactorsFromSymetryAndSMatrices}. 
\begin{lemma}\cite{KarowskiWeiszFormFactorsFromSymetryAndSMatrices}

Let $\mc{F}^{(O)}_{2}(\be_1,\be_2)$ solve $\mathrm{i)-iv)}$ at $n=2$. Then, there exists $k\in \intn{ 0 }{ \tf{\ell}{2} }$, $\varkappa_a\in \Cx$, $a=1,\dots, k$,  such that 
\beq
\mc{F}^{(O)}_{2}(\be_1,\be_2) \, = \, \mc{N}_{\op{O}} \, \pl{a=1}{k} \Big\{ \sinh\Big[ \f{ \be_{12} - \varkappa_a }{2 }\Big] \cdot \sinh\Big[ \f{ \be_{12} + \varkappa_a }{2 }\Big]  \Big\} \, 
\ex{ \f{ \op{s}_{\op{O}} }{ 2 } (\be_1+\be_2) } \,  \op{F}(\be_{12})
\enq
for some $\mc{N}_{\op{O}} \in \Cx$ and where $\op{F}$ is given by the integral representation valid for $0< \Im(\be) < 2\pi$:
\beq
\op{F}(\be) \; = \; \exp\Bigg\{ - 4  \Int{0}{+\infty} \dd x \f{ \sinh(x \mf{b}) \cdot \sinh(x \hat{\mf{b}} ) \cdot \sinh(\tfrac{1}{2}x)   }{ x \sinh^2(x)  }   \cos\Big(  \tfrac{ x  }{ \pi }(\ii\pi - \be)  \Big) \Bigg\} 
\quad \e{with} \quad \hat{\mf{b}} \, = \, \f{1}{2}\, - \, \mf{b} \;. 
\label{ecriture rep int pour F}
\enq

\end{lemma}

\begin{proof}

Axiom $\mathrm{iv)}$ ensures that 
$\mc{F}^{(O)}_{2}(\be_1,\be_2) \, = \,  \ex{ \f{ \op{s}_{\op{O}} }{ 2 } (\be_1+\be_2) } \,  \wt{\op{F}}(\be_1-\be_2)$
for some function $\wt{\op{F}}(\be)$ that is holomorphic on the strip $0<\Im(\be)<2\pi$, bounded at infinity by $C \cosh(\ell \be)$, 
and such that $\wt{\op{F}}_-(\be+2\ii\pi)=\op{S}(\be) \wt{\op{F}}_{+}(-\be)=\wt{\op{F}}_{+}(\be)$, $\be \in \R$.  
One first looks for a particular solution to this scalar Riemann-Hilbert problem, namely a holomorphic function $\op{F}$ in the strip $0 < \Im(\be) < 2\pi $ 
which behaves as $\op{F}(\be) = 1 \, + \, \e{O}\big( \be^{-2} \big)$ as 
$\Re(\be) \rightarrow \pm \infty$ uniformly in $0\leq \Im(\be) \leq 2\pi $ and satisfies $\op{F}_-(\be+2\ii\pi)=\op{S}(-\be) \op{F}_{+}(\be) = \op{F}_{+}(-\be)$, $\be \in \R$.

Starting from the below integral representation  
\beq
 \op{S}(\be) \; = \; \exp\Bigg\{ 8 \Int{0}{+\infty} \dd x \f{ \sinh(x \mf{b}) \cdot \sinh(x \hat{\mf{b}} ) \cdot \sinh(\tfrac{1}{2}x)   }{ x \sinh(x)  }   \sinh\Big(  \tfrac{ x \be }{ \ii \pi } \Big) \Bigg\}
 \;,
\label{ecriture rep int fct S}
\enq
one readily checks that the solution is provided by the below $2\ii\pi$-periodic Cauchy transform
\beq
 \op{F}(\be)  \; = \; \exp\Bigg\{  \Int{\R}{} \f{ \dd s }{ 4\ii\pi }  \coth\Big[ \tfrac{1}{2}(s-\be) \Big] \ln \op{S}(s)   \Bigg\} \;. 
\enq
The $s$ integral can then be taken by means of the integral representation \eqref{ecriture rep int fct S} for $\ln \op{S}(s)$ and leads to \eqref{ecriture rep int pour F}. 
Now it is easy to check that the holomorphic function $\op{G}(\be) = \tf{ \wt{\op{F}}(\be) }{ \op{F}(\be) }$ on the strip $0< \Im(\be) < 2\pi$ admits $\pm$ boundary values and 
satisfies $\op{G}_-(\be+2\ii\pi)=\op{G}_{+}(\be)$ for $\be \in \R$ and is bounded by $C \cosh\big(\ell \Re(\be) \big)$ as $\Re(\be) \rightarrow \infty$ in this strip. As a consequence, it admits a unique
extension into a $2\ii\pi$ periodic entire function bounded by $C \cosh(\ell \be)$ and hence is of the form $P_{\ell}(\ex{\be})$, where $P_{\ell}$ is a Laurent polynomial of maximal positive and negative degree $\ell$. 
Since it is $2\ii\pi$ periodic and even, $P_{\ell}(\ex{\be})$ necessarily takes the form 
\beq
P_{\ell}(\ex{\be})\; = \; \pl{a=1}{k} \Big\{ \sinh\Big[ \f{ \be - \varkappa_a }{2 }\Big] \cdot \sinh\Big[ \f{ \be + \varkappa_a }{2 }\Big]  \Big\} \quad \e{for} \; \e{some} \quad  2k \leq \ell \;. 
\enq

\end{proof}

\subsection{The n-particle sector solution}

 \begin{prop}
 Consider the change of unknown functions
 \beq
 \mc{F}_{n}^{(\op{O})}(\bs{\be}_n) \, = \, \pl{a<b}{n} \op{F}\big(\be_{ab}\big) \cdot \mc{K}_{n}^{(\op{O})}\big( \bs{\be}_{n} \big)              \qquad  with \qquad \be_{ab} \, = \,  \be_{a}-\be_{b} \;, 
 \enq
 with $\op{F}$ as defined through \eqref{ecriture rep int pour F}. Then $ \mc{F}_{n}^{(\op{O})}$ solves the bootstrap axioms $\mathrm{i)-iv)}$ if any only if
\begin{itemize}
 
 \item[I)] $\mc{K}_{n}^{(\op{O})}$  is a symmetric function of $\bs{\be}_n$;

 \item[II)] $\mc{K}_{n}^{(\op{O})}$  is a $2\ii\pi $ periodic and meromorphic function of each variable taken singly;

 \item[III)] the only poles of $\mc{K}_{n}^{(\op{O})}$ are simple and located at $\be_a-\be_b\in \ii\pi (1+2\mathbb{Z})$. The associated residues are given by 
\beq
 \e{Res}\bigg(\mc{K}_{n}^{(\op{O})}(\bs{\be}_n) \cdot \dd \be_1 \, , \, \be_{12} =  \ii \pi  \bigg)  
 \; = \;  \f{\ii }{ \op{F}(\ii\pi) }  \cdot 
   \f{  1\, - \, \pl{a=3}{n}  \op{S}(\be_{2a})  }{  \pl{a=3}{n} \big\{ \op{F}(\be_{2a}+\ii\pi) \op{F}( \be_{2a}) \big\} } 
\cdot 
\mc{K}_{n-2}^{(\op{O})}( \bs{\be}_n^{\prime\prime} )
\enq
where  $\bs{\be}_n^{\prime\prime} \, = \, \big( \be_{3},\dots, \be_n\big)$;
 
 \item[IV)] $ \mc{K}_{n}^{(\op{O})}(\bs{\be}_n + \theta \ov{\bs{e}}_n)  \; = \; \ex{\theta \op{s}_{\op{O}} } \cdot \mc{K}_{n}^{(\op{O})}(\bs{\be}_n  )$.
\end{itemize}

\end{prop}

%
%
%
%
%
%
%
%
This first transformation simplifies the symmetry properties of the problem. However, the inductive reductions provided by the computation of the residues 
are still quite intricate. The idea is then to proceed to yet another change of unknown function, this time by means of a more involved transform. The latter will then lead to structurally much 
simpler, and thus easier to solve, equations satisfied by the new unknown function. As already mentioned, the dawn of this approach goes back to 
\cite{LukyanovConjectureFFExponentialFieldSineGSolASolAndBReather,BrazhnikovLukyanovFreeFieldRepMassiveFFIntegrable} and it was put in the present form in 
\cite{BabujianFringKarowskiZapletalExactFFSineGordonBootsstrapI,BabujianKarowskiExactFFSineGordonBootsstrapII,BabujianKarowskiBreatherFFSineGordon}. 
In particular, we refer to \cite{BabujianKarowskiExactFFSineGordonBootsstrapII}
for the proof.

\begin{prop} \cite{BabujianKarowskiExactFFSineGordonBootsstrapII}

 Let $\bs{\ell}_n\in \{0,1\}^n$ and $p_{n}^{(\op{O})}\big(\bs{\be}_n\mid \bs{\ell}_n\big)$ be a solution to the below constraints 
\begin{itemize}
 
 \item[a)]  $\bs{\be}_n\mapsto p_{n}^{(\op{O})}\big(\bs{\be}_n\mid \bs{\ell}_n\big)$ is a collection of $2\ii\pi$ periodic holomorphic functions on $\Cx$ that are symmetric in the two sets of variables jointly,  \textit{viz}. 
 for any $\sg \in \mf{S}_n$ it holds $ p_{n}^{(\op{O})}\big(\bs{\be}_n^{\sg}\mid \bs{\ell}_n^{\sg} \big)=  p_{n}^{(\op{O})}\big(\bs{\be}_n\mid \bs{\ell}_n\big)$ 
with $\bs{\be}_n^{\sg}=\big( \be_{\sg(1)}, \dots, \be_{\sg(n)} \big)$;
 
  \item[b)] $p_{n}^{(\op{O})}\big(\be_2+\ii\pi, \bs{\be}_n^{\prime}\mid \bs{\ell}_n\big)\, = \,g(\ell_1,\ell_2) p_{n-2}^{(\op{O})}\big(\bs{\be}_n^{\prime\prime}\mid \bs{\ell}_n^{\prime\prime}\big)
  \, + \, h(\ell_1,\ell_2\mid \bs{\be}_n^{\prime})$
where $h$ does not depend on the remaining set of variables $\bs{\ell}_n^{\prime\prime}$ and 
\beq
g(0,1)\, = \, g(1,0) \, = \, \f{ -1 }{ \sin (2\pi \mf{b} ) \,  \op{F}(\ii\pi)  } \, ; 
\enq

 \item[c)]  $ p_{n}^{(\op{O})}\big( \bs{\be}_n + \theta \ov{\bs{e}}_n \mid \bs{\ell}_n \big)  \; = \; \ex{\theta \op{s}_{\op{O}} } \cdot p_{n}^{(\op{O})}\big(\bs{\be}_n\mid \bs{\ell}_n\big)$.

\end{itemize}

 Then, its $\mc{K}$-transform 
\beq
  \mc{K}_{n}\big[ p_n^{(\op{O})} \big]\big( \bs{\be}_{n} \big) \, =  \hspace{-2mm} \sul{ \bs{\ell}_n \in \{0,1\}^n }{} (-1)^{\ov{\bs{\ell}}_n}
\pl{a<b}{n} \bigg\{ 1 \, - \, \ii \f{ \ell_{ab} \cdot \sin[2\pi \mf{b} ] }{ \sinh(\be_{ab})  }  \bigg\} \cdot p_n^{(\op{O})}\big(\bs{\be}_n\mid \bs{\ell}_n\big)  \;, 
\enq
in which $\ov{\bs{\ell}}_n \, = \, \sul{a=1}{n} \ell_k$, solves $\mathrm{ I)-IV) }$.

\end{prop}

Note that arguments were given in \cite{FeiginLashkevichFreeFieldApproachToDescendents} in favour of some form of bijection between certain classes of solutions to $\mathrm{a)-c)}$ and $\mathrm{I)-IV)}$. 
However, we do stress that, so far, the question whether there does exist a clear cut correspondence between all solutions to $\mathrm{a)-c)}$ and $\mathrm{I)-IV)}$ is still open.

\section{Towards physical observables and the convergence problem}

The resolution of the bootstrap program provides one with the expressions for the integral kernels of certain operators which are candidates 
for the quantum fields of the 1+1 dimensional Sinh-Gordon quantum field theory. However, for this construction to really provide one with the quantum field theory
of interest, one should establish several facts. First of all, the operators so constructed should form an algebra in the sense discussed in Subsection \ref{SousSoussectionOperateursDeBase}. 
By virtue of the translational invariance \eqref{ecriture action adjointe operateur de translation}, this means that, for any $n, m \in \mathbb{N}$, the series 
of multiple integrals arising in the operator product $\op{U}_{\op{T}_{\bs{x}}}^{-1}\op{P}_{n} \op{O}_{1}(\bs{x}) \op{O}_{2}(\bs{0}) \op{P}_m$, where $\op{P}_{k}$
is the orthogonal projector on the $k$-particle Fock space, should converge in the weak sense. Namely, any sufficiently regular functions $  \bs{\alpha}_n \mapsto  f^{(n)}(\bs{\alpha}_n)$
and $\bs{\be}_m \mapsto g^{(m)}( \bs{\be}_m )$ belonging, resp., to a dense subset of $L^2(\R^{n}_{>})$ and  $L^2(\R^{m}_{>})$ and for any $d \in \mc{C}^{\infty}_{\e{c}}(\R^{2})$

\begin{multline}
\sul{\ell \geq 0 }{} \Int{\R^{\ell}_{>} }{} \hspace{-1mm} \f{ \dd^{\ell} \ga }{ (2\pi)^{\ell} } 
\bigg\{ \Int{\R^{n}_{>} }{} \hspace{-1mm}\f{ \dd^{n} \alpha }{ (2\pi)^{n} }  f^{(n)}( \bs{\alpha}_n)   \mc{M}_{n;\ell}^{(O_1)}(\bs{\alpha}_n;\bs{\ga}_{\ell}) \bigg\} 
  \\ 
\times \bigg\{ \Int{\R^2 }{} \dd^2  \bs{x} d(\bs{x}) \pl{a=1}{\ell} \ex{-\ii \bs{p}(\ga_a) \cdot \bs{x} }  \bigg\}  \cdot
\bigg\{ \Int{\R^{m}_{>} }{} \hspace{-1mm} \f{ \dd^{n} \be }{ (2\pi)^{m} } \mc{M}_{\ell;m}^{(O_2)}(\bs{\ga}_{\ell};\bs{\be}_m) g^{(m)}( \bs{\be}_m) \bigg\} \;.  
\end{multline}
should converge. The simplest case corresponds to establishing the convergence of the series of multiple integrals subordinate 
to operator products $\op{P}_0\op{O}_1(\bs{x}) \op{O}_2^{\prime}(\bs{0})\op{P}_0$, \textit{viz}. for $n=m=0$. Since the $0^{\e{th}}$ Fock space is one-dimensional, this exactly amounts to the convergence 
of the series of multiple integrals which represents the two-point generalised function $\big< \op{O}_1(\bs{x}) \op{O}_2^{\prime}(\bs{0})\big>$. 
However, even for this specific instance, proving this property on rigorous grounds remained an open 
problem for a very long time. It has only recently been solved by the author  \cite{KozConvergenceFFSeriesSinhGordon2ptFcts} in the case of space-like separation between the operators, \textit{viz}. $\bs{x}^2 <0$. 
The scheme of  proof of this result will be discussed in Section \ref{Section Proof of convergence}. 
From the proof's structure, it is rather clear that one can build on minor modifications of this method so as to 
establish convergence in the time-like regime, \textit{i.e.} when $\bs{x}^2 >0$, although this has not been done yet. Moreover, 
the combinatorial expressions for the kernels $\mc{M}_{n;m}^{(O)}(\bs{\alpha}_n;\bs{\be}_m)$ in terms of the base form factors $\mc{F}_{p}^{(O)}$, $0\leq p \leq m+n$
indicates that the method outlined below would also allow one to tackle the convergence problem for general multi-point correlation functions. 

Once that the convergence problem is solved in full generality, hence guaranteeing that the operators $\op{O}_{\alpha}(\bs{x})$ do form an algebra, 
one still needs to establish the local commutativity property of the quantum fields which ensures causality of the theory. The method for doing so is now well-established. 
Indeed, under the hypothesis of convergence of the handled series of multiple integrals issuing from the operators products, Kirillov and Smirnov 
showed this property in the Sine-Gordon case in \cite{KirillovSmirnovFirstCompleteSetBootstrapAxiomsForQIFT,KirillovSmirnovUseOfBootstrapAxiomsForQIFTToGetMassiveThirringFF}. 
Their method readily applies to the Sinh-Gordon case. Hence, convergence is the only remaining problem so as to set this construction of quantum field theories
on rigorous grounds.

\subsection{The well-poised series expansion for two-point functions}
 
 First of all, by translation invariance, it is enough to focus on  $\big< \op{O}_1(\bs{x}) \op{O}_2(\bs{0})\big>$. Recall that, at least in principle, this quantity is a generalised function 
and should thus be understood, in the first place, as the formal integral kernel of the distribution  $\big< \op{O}_1 \op{O}_2\big>$.  So that, for $d\in \mc{C}^{\infty}_{\e{c}}(\R^2)$, provided convergence holds, one has
\beq
\big< \op{O}_1 \op{O}_2\big>[d] \; = \; \Int{ \R^2 }{}\dd^2 \bs{x} d(\bs{x}) \big< \op{O}_1(\bs{x}) \op{O}_2(\bs{0})\big> \; = \; \sul{ n \geq 0 }{} \f{ 1 }{ n! } \mc{I}_n^{(\op{O}_1,\op{O}_2)}[d]
\enq
with 
\beq
\mc{I}_n^{(\op{O}_1,\op{O}_2)} [d] \, = \, \Int{  \R^n }{} \f{ \dd^n \be }{ (2\pi)^n } 
 \mc{F}_{n}^{(\op{O}_1)}(\bs{\be}_n)   \mc{M}_{n;0}^{(\op{O}_2)}(\bs{\be}_{n}; \emptyset ) 
 \Int{ \R^2 }{}\hspace{-1mm} \dd^2 \bs{x} d(\bs{x})   \pl{a=1}{n}  \Big\{ \ex{-\ii m [ t \cosh(\be_a) - x \sinh(\be_a)] }\Big\}   \;. 
\nonumber
\enq
It is a direct consequence of the kernel reduction axiom $\mathrm{v)}$ and of Lorentz invariance one $\mathrm{iv)}$ that 
\beq
\mc{M}_{n;0}^{(\op{O}_2)}(\bs{\be}_{n}; \emptyset )  \, = \, \mc{F}_{n}^{(\op{O}_2)}( \overleftarrow{\bs{\be}}_n+\ii\pi \ov{\bs{e}}_n )  \, = \, 
\ex{\ii \pi \op{s}_{\op{O}_2} } \mc{F}_{n}^{(\op{O}_2)}( \overleftarrow{\bs{\be}}_n  ) \;. 
\enq
This identity along with the growth bounds in each $\be_a$ of the 
 form factors  $\mc{F}_{n}^{(\op{O})}(\bs{\be}_n)$ ensures the well definiteness of the $n$-fold integrals since the space-time integral over $\bs{x}$ produces 
a decay in each $\be_a$ that is faster than any exponential $\ex{ \pm k \be_a }$, $\Re(\be_a)\rightarrow \pm \infty$. 
By virtue of the Morera theorem, this rapid decay at infinity along with the holomorphic properties of the integrands allow one to deform, \textit{simultaneously} for each integration variable $\be_a$, $a=1,\dots, n$ 
the integration curves to  $\R+\ii\tfrac{\pi}{2}\e{sgn}(x)$ when $\bs{x}$ is space-like and, when $\bs{x}$ is time-like, to $\ga(\R)$ where $\ga(u)=u+\ii\vartheta(u)$, 
where $\vartheta$ is smooth, $|\vartheta|<\tf{\pi}{4}$, and such that there exists $M>0$ large enough and $0<\eps < \tf{\pi}{2}$  so that $\vartheta(u)=-\e{sgn}(t) \e{sgn}(u) \eps$ when $|u|\geq M$.
This operation turns the $\bs{\be}_n$ integrals into absolutely convergent ones irrespectively of the presence of $d(\bs{x})$. In particular, for the space-like 
regime, one gets that 
\beq
\mc{I}_n^{(\op{O}_1,\op{O}_2)}  =  \ex{ \eta(x) }  \hspace{-2mm} \Int{ \R^2 }{}  \hspace{-1mm} \dd^2 \bs{x} d(\bs{x})   \hspace{-1mm} \Int{  \R^n }{}  \hspace{-1mm} 
\f{ \dd^n \be }{ (2\pi)^n } 
 \mc{F}_{n}^{(\op{O}_1)}(\bs{\be}_n)    \mc{F}_{n}^{(\op{O}_2)}( \overleftarrow{\bs{\be}}_n  ) 
   \pl{a=1}{n}     \ex{- m r   \cosh(\be_a)  }        \;,
\nonumber 
\enq
in which  $r = \sqrt{x^2-t^2}$,  $\tanh(\vartheta)=\tf{t}{x}$ while $\eta( x )=\ii \pi \op{s}_{\op{O}_2} \, + \, (\ii \tfrac{ \pi }{ 2 }  + \vartheta )(\op{s}_{\op{O}_1}+\op{s}_{\op{O}_2})\e{sgn}(x) $. Hence, provided convergence holds, one has the well-defined in the usual sense of numbers representation for the two-point function
\beq
\big< \op{O}_1(\bs{x}) \op{O}_2(\bs{0})\big> \; = \; \ex{ \eta(x) }  \sul{ n \geq 0 }{} \f{ 1 }{ n! }\Int{  \R^n }{}  \hspace{-1mm} 
\f{ \dd^n \be }{ (2\pi)^n } 
 \mc{F}_{n}^{(\op{O}_1)}(\bs{\be}_n)    \mc{F}_{n}^{(\op{O}_2)}( \overleftarrow{\bs{\be}}_n  ) 
   \pl{a=1}{n}     \ex{- m r   \cosh(\be_a)  } 
\label{ecriture developpement serie fct 2 pts}
\enq

\subsection{Convergence of series representation for two-point functions}

Thus, the well-definiteness of the two-point functions boils down to providing an appropriate upper bound for the below class of $N$-fold integrals for $\varkappa>0$
\beq
\mc{Z}_N(\varkappa)   =  \Int{  \R^N }{}  \hspace{-2mm}  \dd^N \! \be   \pl{a\not= b }{N }  \ex{ \f{1}{2} \mf{w}(\be_{ab}) }  \cdot  \pl{a=1}{N}  \Big\{ \ex{- 2 \varkappa  \cosh(\be_a) }\Big\} 
\mc{K}_{N}\big[p^{(\op{O}_1)}_N \big](\bs{\be}_N) \mc{K}_{N}\big[p^{(\op{O}_2)}_N \big]( \overleftarrow{\bs{\be}}_N)\;. 
\label{definition fonction de partition de depart}
\enq
The  two-body potential $\mf{w}$ is defined through  the relation $\op{F}(\la)\op{F}(-\la) \; = \; \ex{   \mf{w} (\la) } $. 

\begin{theorem}\cite{KozConvergenceFFSeriesSinhGordon2ptFcts}
 \label{Theorem borne sup sur fct partition}
 Assume that there exist $C_1, C_2$ and $k\in \mathbb{N}$ such that given $s\in \{1,2\}$
\beq
\big| p^{(\op{O}_s)}_N(\bs{\be}_N\mid \bs{\ell}_N  ) \big| \, \leq \, C_1^N \cdot \pl{a=1}{N}\ex{ C_2 \be_a^{k}} \quad for \; any \quad \bs{\ell}_{N} \in \{0,1\}^N
\enq
uniformly  in $N$. Then, it holds 
\beq
\big|  \mc{Z}_N(  \varkappa) \big| \, \leq \, \exp\bigg[ - \f{3\pi^2 \mf{b} \hat{\mf{b}} \cdot N^2 }{ 4 \cdot (\ln N)^3 }  \Big\{  1 \, + \, \e{O}\Big( \f{1}{\ln N } \Big) \Big\} \bigg]
\enq

\end{theorem}

The proof of this theorem was the goal of the author's work \cite{KozConvergenceFFSeriesSinhGordon2ptFcts}. The proof relies on Riemann--Hilbert problem techniques for inverting singular integral operators
of truncated-Wiener Hopf type along with the Deift-Zhou non-linear steepest descent method \cite{DeiftZhouSteepestDescentForOscillatoryRHP,DeiftZhouSteepestDescentForOscillatoryRHPmKdVIntroMethod}, 
 concentration of measure and large deviation techniques
which were developed for dealing with certain 
$\be$-ensembles multiple integrals \cite{BoutetdeMonvelPasturShcherbinaAPrioriBoundsOnFluctuationsAroundEqMeas,BenArousGuionnetLargeDeviationForWignerLawLEadingAsymptOneMatrixIntegral,
MaidaMaurelSegalaInegalitesPourConcentrationDeMesures}, and 
some generalisations thereof to the case of $N$-dependent integrands in $N$-dimensional integrals as it was developed in \cite{KozBorotGuionnetLargeNBehMulIntOfToyModelSoVType}.

\section{The proof of the convergence of the form factor series}
\label{Section Proof of convergence}
In this section we shall describe the main steps of the proof. The details can be found in Proposition 3.1 of \cite{KozConvergenceFFSeriesSinhGordon2ptFcts}.

\subsection{An simpler upper bound}

The starting point consists in obtaining a structurally simpler upper bound on $ \mc{Z}_N(  \varkappa) $ when $\varkappa>0$. 

\begin{prop}
 There exists C>0 such that 
\beq
\big|  \mc{Z}_N(  \varkappa) \big| \, \leq \,  \Big( C  \cdot  \ln N  \Big)^N \cdot  \e{max}_{p\in \intn{0}{N}} \big| \msc{Z}_{N,p}(\varkappa) \big| 
\enq
 where $\msc{Z}_{N,p}( \varkappa )  \, = \,    \Int{\R^{N-p} }{} \hspace{-2mm}  \dd^{N-p} \la    \Int{  \R^p }{} \hspace{-1mm} \dd^p \nu  \;  \ov{\varrho}_{N,p}\big(\bs{\la}_{N-p}, \bs{\nu}_p \big)$
whose integrand is expressed as 
\begin{multline}
 \ov{\varrho}_{N,p}\big(\bs{\la}_{N-p}, \bs{\nu}_p \big) \; = \; 
\pl{a=1}{p}  \Big\{ \ex{- V_N(\nu_a) }\Big\} \cdot  \pl{a=1}{N-p}  \Big\{ \ex{-  V_N(\la_a) }\Big\}  \\
\times \pl{a< b }{p } \Big\{ \ex{  \mf{w}_N(\nu_{ab}) } \Big\}   
\cdot \pl{a< b }{ N - p } \Big\{ \ex{  \mf{w}_N(\la_{ab}) } \Big\}  \cdot \pl{a=1}{p} \pl{b=1}{N-p}   \bigg\{  \ex{  \mf{w}_{\e{tot};N}(\nu_{a}-\la_{b}) }    \bigg\}  \;. 
\label{definition densite integrande de Z N p majorant}
\end{multline}
Above, there appears the $N$-dependent functions
\beq
V_N(\la) \, = \, \varkappa \cosh(\tau_N \la) \quad , \quad  \mf{w}_{N}(\la) \, = \, \mf{w}(\tau_N \la) \quad, \quad  \mf{w}_{\e{tot};N}(\la) \, = \, \mf{w}_{\e{tot}}(\tau_N \la) \;, 
\enq
with $\tau_N=\ln N$ and one has 
\beq
 \mf{w}_{\e{tot}}(\la)  \, = \, \mf{w}(\la) \, + \, \mf{v}_{2\pi \mf{b}, 0^+}(\la) \quad with \quad 
\mf{v}_{ \alpha , \eta}(\la) \, = \, \ln \Bigg(  \f{  \sinh(\la + \ii \alpha) \, \sinh(\la - \ii\alpha) }{  \sinh(\la + \ii \eta) \, \sinh(\la - \ii \eta ) }   \Bigg) \;. 
\label{definition potentiel varpi tot et correctif v alpha eta}
\enq

\end{prop}

\subsection{Energetic bounds}

\begin{prop}
 
\label{Proposition estimation comportement gd N Z N p majorant}

The partition function $\msc{Z}_{N,p}( \varkappa )$ admits the upper bound:
\beq
\msc{Z}_{N,p}( \varkappa )  \,  \leq \,  \exp\bigg\{ -N^2   \e{inf} \Big\{ \,  \mc{E}_{N,\frac{p}{N}}[\mu,\nu]  \; : \;  (\mu,\nu)\in  \mc{M}^{1}(\R)\times \mc{M}^{1}(\R)   \Big\} \; + \; \e{O}\big( N \tau_N^2 \big)  \bigg\} \,,  
\label{ecriture borne sup sur fct part majorante}
\enq
in which the control is uniform in $p\in \intn{0}{N}$ and where  
\begin{multline}
\nonumber 
\mc{E}_{N,t}[\mu,\nu]\; = \; \f{1}{N} \bigg\{ t \Int{}{}\hspace{-1mm}  V_N(s) \dd \nu(s) \, + \,  (1-t) \Int{}{} \hspace{-1mm} V_N(s) \dd \mu(s) \bigg\}  -  
\f{t^2}{2}\!  \Int{}{} \hspace{-1mm}  \mf{w}_N(s-u)  \dd \nu(s) \dd \nu(u)  \\ 
 \; - \; \f{ (1-t)^2}{2} \Int{}{} \mf{w}_N(s-u)   \dd \mu(s) \dd \mu(u)  \, - \, t(1-t)  \Int{}{} \mf{w}_{\e{tot};N}(s-u)   \dd \mu(s) \dd \nu(u)   \;. 
\label{definition fonctionelle E N t}
\end{multline}
\end{prop}

One may obtain such an upper bound within the standard approach to establishing large deviation bounds for $N$-fold integrals as pioneered in \cite{BenArousGuionnetLargeDeviationForWignerLawLEadingAsymptOneMatrixIntegral}, 
adjoined to the local regularisation of the empirical distribution of the integration variables proposed in \cite{MaidaMaurelSegalaInegalitesPourConcentrationDeMesures}
and some fine bounds due to the $N$-dependence of the integrand which were also considered in \cite{KozBorotGuionnetLargeNBehMulIntOfToyModelSoVType}. The details
can be found in Lemmata 4.2-4.3 of \cite{KozConvergenceFFSeriesSinhGordon2ptFcts}.

\subsection{Characterisation of the minimiser and a lower-bound minimiser}

The upper bound established in Proposition \ref{Proposition estimation comportement gd N Z N p majorant} does not allow one to conclude directly on the convergence of the series. Indeed, 
even if one could prove that the infimum in \eqref{ecriture borne sup sur fct part majorante} gives a strictly positive number, the $N$-dependence of the energy functional 
could make the infimum $N$-dependent and, in principle, the latter could give rise to a behaviour in $N$ which, when multiplied by the $N^2$ prefactor, could turn out to be subdominant in respect to 
the corrections $\e{O}(N \tau_N^2)$. Hence, the longest part of the proof is devoted to obtaining some sharp and explicit lower bound for the infimum which can then be computed in closed form
so that one may explicitly check that the above scenario does hold. 

For that purpose,  one starts by showing the

 \begin{prop}
 \label{Proposition minimiseur unique pour mathcal E de N et t} 
  For $0<t<1$ $\mc{E}_{N,t}$ admits a unique minimiser $( \mu_{\e{eq}}^{(N,t)}, \nu_{\e{eq}}^{(N,t)} )$ on $\mc{M}^{1}(\R)\times \mc{M}^{1}(\R)$. Similarly, $\mc{E}_{N,0}$  and $\mc{E}_{N,1}$ 
  admit unique minimisers on $\mc{M}^{1}(\R)$.
 
 \end{prop}

This is established by showing that, for $0<t<1$, $\mc{E}_{N,t}$ is lower semi-continuous and strictly convex on $\mc{M}^{1}(\R)\times \mc{M}^{1}(\R)$, has compact level sets, is not identically $+\infty$ and is bounded from below. 
In principle, this result could be already enough to obtain sharp in $N$ estimates for $\mc{E}_{N,t}[ \mu_{\e{eq}}^{(N,t)}, \nu_{\e{eq}}^{(N,t)} ]$. Indeed, by relying on the analogous to the case of $\be$-ensembles variational characterisation of the minimisers and showing that these are actually Lebesgue continuous with compact connected supports, one may establish a system of two-coupled singular linear integral equations of truncated Wiener-Hopf type depending on the large-parameter $N$. These may be analysed within the method developed by Krein's school after generalising the work \cite{NovokshenovSingIntEqnsIntervalGeneral} 
and solving the $4\times 4$ associated Riemann--Hilbert problem in the large-$N$ regime by the Deift-Zhou non-linear steepest descent method \cite{DeiftZhouSteepestDescentForOscillatoryRHP,DeiftZhouSteepestDescentForOscillatoryRHPmKdVIntroMethod}. However, these steps would definitely lead to an extremely cumbersome and long clamber, especially taken the 
minimal amount of information one needs, in the end, from such handlings. Therefore, it is more convenient to reduce the numbers of minimisers which ought to be thoroughly determined
by providing a lower bound for $\mc{E}_{N,t}\big[ \mu_{\e{eq}}^{(N,t)}, \nu_{\e{eq}}^{(N,t)} \big]$ whose estimation would demand less effort while still leading to the desired result.

A direct calculation shows that one has a simpler representation for $\mc{E}_{N,t}$ in terms of functionals only acting on one copy of a 
space of bounded measures:
\beq
\mc{E}_{N,t}\big[\mu,\nu\big] \; = \; \sul{\ups = \pm }{}   \mc{E}_{N}^{(\ups)}\big[ \sg^{(\ups)}_t \big] \qquad  \e{with}  \qquad 
\sg^{(\pm)}_{t} \, = \, t \nu \pm (1-t) \mu \;,  
\enq
in which $\mc{E}_{N}^{(+)}$ is a functional on $\mc{M}^{1}(\R)$ while $\mc{E}_{N}^{(-)}$  is a functional on $\mc{M}^{(2t-1)}_{\mf{s}}(\R)$, the space of signed, bounded, measures on $\R$ of total mass $2t-1$. 
These take the form 
\begin{eqnarray}
\mc{E}_{N}^{(+)}\big[\sg \big]  & = & \f{ 1 }{ N }   \Int{}{} V_N(s)\,  \dd \sg(s) \, - \, \f{ 1 }{ 2 } \Int{}{} \mf{w}_N^{(+)}(s-u) \cdot \dd \sg(s) \, \dd \sg(u) \, ,  \label{definition fonctionnelle energie +}\\
\mc{E}_{N}^{(-)}\big[ \sg \big]  & = &  - \f{ 1 }{ 2 } \Int{}{} \mf{w}_N^{(-)}(s-t) \cdot \dd \sg(s) \, \dd \sg(u) \;  . 
\end{eqnarray}
The two-body interactions appearing above involve $\mf{w}$ and $\mf{v}_{\alpha,\eta}$ introduced in \eqref{definition fonction de partition de depart} and 
\eqref{definition potentiel varpi tot et correctif v alpha eta} 
\beq
\mf{w}_N^{(\pm)}(u) \; = \; \mf{w}^{(\pm)}(\tau_N u ) \qquad \e{with} \qquad 
\left\{ \begin{array}{ccc}   \mf{w}^{(+)}(u  )& = &   \mf{w}(u) + \f{1}{2} \mf{v}_{2\pi \mf{b},0^+}(u)   \\
  \mf{w}^{(-)}(u ) &  =  & - \f{1}{2}\mf{v}_{2\pi \mf{b},0^+}(u)   \end{array} \right. \;. 
\enq
By going to Fourier space, one observes that 
\beq
\mc{E}_{N}^{(-)}[ \sg ] \; = \;  \f{1}{2} \Int{}{} \hspace{-1mm} \dd \la \, \big| \mc{F}[ \sg ](\la)\big|^2  \f{ \sinh( \pi \mf{b} \la ) \cdot \sinh(\pi \hat{\mf{b}} \la )   }{ \la \sinh(\tfrac{\pi}{2}\la)    }  
\; \geq \; 0 \, , 
\enq
where $\mc{F}[ \sg ](\la)$ stands for the Fourier transform of the signed measure $\sg$. Thus, 
\beq
\mc{E}_{N,t}[ \mu_{\e{eq}}^{(N,t)}, \nu_{\e{eq}}^{(N,t)} ] \; \geq \; \mc{E}_{N}^{(+)}\big[t \nu_{\e{eq}}^{(N,t)} + (1-t) \mu_{\e{eq}}^{(N,t)}   \big]  \, \geq \, \mc{E}_{N}^{(+)}\big[ \sg_{\e{eq}}^{(N)} \big]
  \;. 
\label{ecriture borne in sur minimum E n t avec minimum de EN plus}
\enq
  In the last line, we have used that $\mc{E}_{N}^{(+)}$ is lower-continuous, has compact level sets, is strictly convex on $\mc{M}^{1}(\R)$, bounded from below and not identically $+\infty$
so as to ensure the existence of a unique minimiser thereof: 
$\mc{E}_{N}^{(+)}\big[ \sg_{\e{eq}}^{(N)} \big] \; = \;  \e{inf} \Big\{   \mc{E}_{N}^{(+)}\big[ \sg \big] \; : \;   \sg    \in \mc{M}^{1}(\R) \Big\}$.

\subsection{Singular integral equation characterisation of the minimiser $ \sg_{\e{eq}}^{(N)}$}

By using the variational characterisation of the minimiser, see \textit{e.g.} \cite{DeiftOrthPlyAndRandomMatrixRHP}  for an exposition in the $\be$-ensemble case, 
one reduces the construction of $\sg_{\e{eq}}^{(N)}$ to finding a solution to a singular integral equation  on the Sobolev space $H_s(\intff{a_N}{b_N})$ driven by the operator 
\beq
\mc{S}_N\big[ \phi \big](\xi) \; = \;  \Fint{a_N}{b_N} \big( \mf{w}^{(+)}\big)^{\prime}\big[ \tau_N (\xi-\eta) \big]  \cdot   \phi(\eta) \dd \eta \; . 
\enq
Indeed, upon introducing the effective potential subordinate to a function $\phi\in H_s(\intff{a_N}{b_N})$
\beq
V_{N;\e{eff}}[\phi](\xi) \; = \; \f{1}{N} V_N(\xi) \, - \, \Fint{a_N}{b_N} \mf{w}^{(+)}\big[ \tau_N (\xi-\eta) \big]  \cdot   \phi(\eta) \dd \eta 
\label{ecriture potentiel effectif associee a une solution phi}
\enq
one may formulate the 
\begin{prop}
\label{Proposition characterisation qqes ptes mesure eq}
Let $a_N<b_N$ and $\varrho_{\e{eq}}^{(N)}\in H_s(\intff{a_N}{b_N})$, $\tf{1}{2}<s<1$ solve
\beq
\f{ 1 }{ N \tau_N } V_N^{\prime}(x)  \; = \; \mc{S}_N\big[ \varrho_{\e{eq}}^{(N)} \big](x)  \qquad  \; on \quad \intoo{a_N}{b_N} \; , 
\label{ecriture eqn int sing lin pour densite mesure eq}
\enq
be subject to the conditions 
\beq
\varrho_{\e{eq}}^{(N)}(\xi) \geq  0 \quad for \quad \xi \in \intff{a_N}{b_N} \; , \qquad \Int{a_N}{b_N}\varrho_{\e{eq}}^{(N)}(\xi) \dd \xi \; = \; 1 
\label{ecriture masse A densite}
\enq
and 
\beq
V_{N;\e{eff}}[ \varrho_{\e{eq}}^{(N)} ](\xi) \; > \; \e{inf}\big\{  V_{N;\e{eff}}[ \varrho_{\e{eq}}^{(N)} ](\eta) \, : \, \eta \in \R \big\} \quad for \; any \;\; \xi \in \R \setminus \intff{a_N}{b_N} \;. 
\label{ecriture positivite stricte du potentiel effectif}
\enq
Then, the equilibrium measure $\sg_{\e{eq}}^{(N)}$ is supported on the segment $\intff{a_N}{b_N}$ and continuous in respect to Lebesgue's measure with density $ \varrho_{\e{eq}}^{(N)}  $. 
Moreover,  the density takes the form 
\beq
 \varrho_{\e{eq}}^{(N)}(\xi) \, = \, \sqrt{(b_N-\xi)(\xi-a_N)} \cdot h_N(\xi) \quad with \quad h_N \in \mc{C}^{\infty}(\intff{a_N}{b_N}) \,. 
\label{ecriture forme densite mesure equilibre}
\enq

\end{prop}

The above proposition thus provides one with the following strategy for determining the equilibrium measure. One starts by solving 
the singular integral equation \eqref{ecriture eqn int sing lin pour densite mesure eq} for \textit{any} endpoints $a_N$ and $b_N$.  
The inversion should be carried out in an appropriate functional space which is dictated by the local structure \eqref{ecriture forme densite mesure equilibre} of the equilibrium 
measure's density, as can be inferred from an analysis of the systems of loop equations  
associated with the probability measure on $\R^N$ naturally subordinate to the energy functional $\mc{E}_{N}^{(+)}$. 
The fact that $\mc{S}_N$ should be inverted on $H_s(\intff{a_N}{b_N})$, $0<s<1$, imposes a constraint on $a_N$ and $b_N$. A second constraint
is obtained from the fact that the equilibrium measure has mass one \eqref{ecriture masse A densite}. 
This is still not enough so as to be sure that the solution constructed in this way provides one with the equilibrium measure. 
For that to happen, one still needs to verify that the two positivity constraints \eqref{ecriture masse A densite}-\eqref{ecriture positivite stricte du potentiel effectif}
are fulfilled. The realisation of such a program demands to have a thorough control on the inversion of $\mc{S}_N$. 
The latter may be reached within the scheme developed in \cite{NovokshenovSingIntEqnsIntervalGeneral}, by solving an auxiliary $2\times 2$
Riemann--Hilbert problem.

\subsection{The Riemann--Hilbert based inversion of the operator}

In the following, we adopt the shorthand notations
\beq
\ov{a}_N=\tau_N a_N \; , \quad \ov{b}_N=\tau_N b_N \; , \quad \ov{x}_N = \tau_N(b_N-a_N) \;. 
\enq
Consider the Riemann--Hilbert problem for a $2\times 2$ matrix function $\chi\in \mc{M}_2\big( \mc{O}(\Cx\setminus\R) \big)$:
\begin{itemize}

\item $\chi$ has continuous $\pm$-boundary values on  $\R$;
\item there exist constant matrices $\chi^{(a)}$ with  $\chi^{(1)}_{12}\not= 0$ such that  when  $\la \rightarrow \infty$ 
\beq
{\small
\hspace{-1.5cm}
\chi(\la) = \left\{ 
\begin{array}{cc} \mc{P}_{L;\uparrow}(\la) \cdot  \left( \begin{array}{cc} -\mf{s}_{\la} \cdot \ex{ \ii\la \ov{x}_N }  & 1 \\
										-1 &  0     \end{array} \right)  
		      \cdot \f{ \big(-\ii \la \big)^{  \f{3}{2} \sg_3 } }{ \ex{ -\ii\f{3 \pi}{2}\sg_3}  } \cdot 
  \Big(I_{2}    +   \f{\chi^{(1)}}{\la}   +   \f{\chi^{(2)}}{\la^2}   +   \e{O}\big(\la^{-3}\big) \Big) \cdot \op{Q}(\la)  & 
			 \la \in \mathbb{H}^+ \vspace{3mm} \\
\mc{P}_{L;\downarrow}(\la) \cdot  \left( \begin{array}{cc} -1  & \mf{s}_{\la} \cdot \ex{- \ii\la \ov{x}_N }   \\
										0  & 1    \end{array} \right)
			  \cdot \big( \ii \la \big)^{\f{3 }{2} \sg_3 }  \cdot 
		     \Big(I_{2}    +   \f{\chi^{(1)}}{\la}   +   \f{\chi^{(2)}}{\la^2}   +   \e{O}\big(\la^{-3}\big) \Big) \cdot \op{Q}(\la) &
				\la \in \mathbb{H}^-  \end{array}\right. 
}
\nonumber 
\enq

in which the matrix $\op{Q}$ takes the form

\vspace{2mm}
$ \op{Q}(\la) \; = \; \left( \begin{array}{cc}  0   & - \chi^{(1)}_{12}  \\ 
   \Big\{  \chi^{(1)}_{12} \Big\}^{-1}   &   \mf{q}_1 \, + \, \la   \end{array} \right)  \quad \e{with} \quad \mf{q}_1 \, = \, \Big( \chi^{(1)}_{11} \chi^{(1)}_{12} \, - \, \chi^{(2)}_{12}  \Big)\cdot  \Big\{  \chi^{(1)}_{12} \Big\}^{-1}  $  \;;

\item $\chi_+(\la) \; = \;   G_{\chi}(\la) \cdot \chi_-(\la) $ \quad for \quad  $\la \in \R$  where
\beq
G_{\chi}(\la) \; = \; \left( \begin{array}{cc}   \ex{ \ii \la \ov{x}_N }  & 0  \\  
							\frac{ 1 }{  \ii \pi  } \cdot R(\la)    &  -\ex{- \ii \la \ov{x}_N }  \end{array} \right)  \quad \e{with} \quad R(\la) \; = \; 2 \f{ \sinh( \pi \mf{b} \la ) \cdot \sinh(\pi \hat{\mf{b}} \la ) \cdot \sinh(\tfrac{\pi}{2}\la)   }{  \cosh^2(\tfrac{\pi}{2} \la)   }  \;. 							 
\nonumber
\enq
\end{itemize}
Here $\mf{s}_{\la}=\e{sgn}\big( \Re\,\la \big)$, $\mc{O}(A)$ stands for the ring of holomorphic functions on $A$, while the $\e{O}$ remainder appearing in matrix equalities should be understood entrywise. Moreover, 
we point out that the matrix $\op{Q}$ appearing in the asymptotic expansion for $\chi$ is chosen such that $\chi$
has the large-$\la$ behaviour 
\beq
\chi(\la) \, = \, \chi_{\uparrow/\downarrow}^{(\infty)}(\la) \cdot  \big(\mp\ii \la \big)^{  \f{1}{2} \sg_3 }\qquad \la \in \mathbb{H}^{\pm} \;, 
\enq
with $\chi_{\uparrow/\downarrow}^{(\infty)}(\la)$ bounded at $\infty$.

The Deift-Zhou non-linear steepest descent method \cite{DeiftZhouSteepestDescentForOscillatoryRHP,DeiftZhouSteepestDescentForOscillatoryRHPmKdVIntroMethod}
allows one to reduce the above Riemann-Hilbert problem into one that is uniquely solvable by the singular integral equation method of \cite{BealsCoifmanScatteringInFirstOrderSystemsEquivalenceRHPSingIntEqnMention}, 
provided that $N$ is large enough and $b_N-a_N>c>0$ uniformly in $N$. 

The solution $\chi$ then provides one with a full description of the inverse of $\mc{S}_N$.
\begin{prop}
\label{Proposition invertibilite operateur S}

Let $0 < s < 1$. The operator 
$\mc{S}_N \; : \; H_{s}\big( \intff{a_N}{b_N} \big) \longrightarrow H_{s}\big( \R \big)$
is continuous and invertible on its image:
\beq
\label{x}
\mf{X}_{s}\big( \R \big) = \Bigg\{H \in H_{s}(\R)\;\; :\;    \Int{\R+\ii \eps^{\prime} }{} \hspace{-2mm} \chi_{12}(\mu) \, \mc{F}[H](\tau_N\mu)\, \ex{- \ii \mu \ov{b}_N}\cdot\f{\dd\mu}{ (2 \ii \pi)^2 }  \; = \; 0\Bigg\}\;.
\enq 
More specifically, one has the left and right inverse relations 
\beq
\mc{W}_N\circ \mc{S}_N \, = \, \e{id}  \quad on \quad H_s(\intff{a_N}{b_N}) \quad and \quad 
\mc{S}_N\circ \mc{W}_N[H](\xi) \, = \, H(\xi)   \quad a.e. \; on\;  \intff{a_N}{b_N}
\nonumber
\enq
for any $H \in \mf{X}_s(\R)$. The operator $\mc{W}_N\,:\,\mf{X}_{s}(\R) \longrightarrow H_s(\intff{a_N}{b_N})$ is given, whenever it makes sense, as an encased oscillatorily convergent Riemann integral transform
\beq
\label{forumule explicite pour WN}
\mc{W}_N[H](\xi) \; = \; \f{ \tau_N^{2} }{ \pi }  \hspace{-2mm}
\Int{ \R + 2 \ii \eps^{\prime} }{} \hspace{-2mm} \f{ \dd \la }{ 2 \ii \pi } \Int{ \R + \ii \eps^{\prime} }{}\hspace{-2mm} \f{ \dd \mu }{ 2 \ii \pi }\, 
 \ex{- \ii \tau_N\la(\xi-a_N) } W(\la,\mu)
\,\ex{- \ii  \mu \ov{b}_N} \mc{F}[H](\tau_N \mu ) \;, 
\enq
 where $\eps^{\prime}> 0 $ is small enough. 
The integral kernel 
\beq
W(\la,\mu) \, = \, \f{ 1 }{ \mu- \la }
\bigg\{    \f{ \mu }{ \la }  \cdot \chi_{11}(\la) \chi_{12}(\mu) \, -  \, \chi_{11}(\mu) \chi_{12}(\la) \bigg\} \;, 
\label{definition noyau W}
\enq
is expressed in terms of the entries of the matrix  $\chi$. 
\end{prop}

These pieces of information along with the explicit, uniform on $\Cx$, large-$N$ expansion of the solution $\chi$  to the above Riemann--Hilbert problem 
and several technical estimates which allow one to check that \eqref{ecriture masse A densite}-\eqref{ecriture positivite stricte du potentiel effectif} hold,
allow one to formulate the

\begin{theorem}
\label{Theorem caracterisation minimiseur fnelle EN +}
 Let $N\geq N_0$ with $N_0$ large enough. Then the unique minimiser $\sg_{\e{eq}}^{(N)}$ of the functional $\mc{E}_N^{(+)}$ 
 introduced in \eqref{definition fonctionnelle energie +} is absolutely continuous in respect to the Lebesgue measure with density $\varrho_{\e{eq}}^{(N)}$
 and is supported on the segment  $\intff{a_N}{b_N}$. The endpoints are the unique solutions to the equations 
\beq
a_N+b_N=0 \qquad and \qquad    \vartheta  \cdot \f{    ( \ov{b}_N)^2   \, \ex{ \ov{b}_N }  }{    N  } \cdot \mf{t}( 2 \ov{b}_N) \cdot
\Big\{ 1   \, + \, \e{O}\Big(  (\ov{b}_N)^5 \ex{ -2\ov{b}_N(1-\eps)  }   \Big) \Big\} \, = \, 1 \;, 
\nonumber
\enq
for any  $1>\eps>0$  and the remainder is smooth and differentiable in $\ov{b}_N$. Above, one has 
\beq
\vartheta \, = \, \f{ 2 \, \varkappa }{ 3  (2\pi)^{ \f{5}{2} } } \cdot \f{ \Ga\big( \mf{b}, \hat{\mf{b}} \big) }{ \mf{b}^{\mf{b}} \, \hat{\mf{b}}^{\hat{\mf{b}}} }\; , 
\nonumber
\enq
while, upon using the constants $w_k$ introduced below in \eqref{definition des constantes wk et de leur asymptotiques},  
\beq
\mf{t}( \ov{x}_N) \, = \,  \f{ 6 }{ (\ov{x}_N)^2 } \Big\{  2  \, + \, w_2  \, - \, w_1 \, -\, \f{w_1 w_3}{ w_2 }  \Big\} \; \widesim{  \ov{x}_N \rightarrow + \infty } \; 1 \, + \, \e{O}\Big( \tfrac{ 1 }{ \ov{x}_N } \Big) \;. 
\label{definition polynome frak t}
\enq
In particular,  $\ov{b}_N$ is uniformly away from zero and admits the large-$N$ expansion 
\beq
\ov{b}_N\, = \, \ln N  \,-  \, 2  \ln \ln N  \, - \,  \ln \vartheta   \; + \; \e{O}\Big( \f{ \ln \ln N }{ \ln N }  \Big) \;. 
\nonumber
\enq

 Finally, the density $\varrho_{\e{eq}}^{(N)}$ of the equilibrium measure  
 is expressed in terms of the integral transform of the potential $\varrho_{\e{eq}}^{(N)} \, = \, \tf{ \mc{W}_N[V^{\prime}_N] }{ (N \tau_N) }$. 
 
\end{theorem}

In the statement of the theorem, we made use of the coefficients $w_k$ arising in the $\la \rightarrow 0$ expansion below 
\beq
 2 \ii   \f{  \mf{b}^{  2  \ii \mf{b} \la}   \hat{\mf{b}}^{  2 \ii \hat{\mf{b}} \la}    2^{ \ii \la }    }{ \la^3 \, \mf{b}  \hat{\mf{b}}   \ex{   \ii \la \ov{x}_N } }       
\Ga^2 \! \left( \hspace{-1mm}\begin{array}{c}  \tfrac{1}{2} + \ii \tfrac{\la}{2}  \vspace{1mm} \\ \tfrac{1}{2} - \ii \tfrac{\la}{2}    \end{array} \right)  
\Ga\! \left( \hspace{-1mm} \begin{array}{c}  1-\ii \mf{b} \la   , 1-\ii \hat{\mf{b}} \la   ,  1-\ii \tfrac{\la}{2}  \vspace{1mm} \\
 \ii \mf{b} \la \, , \ii \hat{\mf{b}} \la \, ,  \ii \tfrac{\la}{2}  \end{array} \right)  
  =  \sul{\ell=0}{3} \f{(-\ii)^{\ell} w_{\ell}  }{ \la^{3-\ell} } \; + \; \e{O}\big( \la \big)  \;.
\label{definition des constantes wk et de leur asymptotiques}
\enq

\subsection{Estimation of the minimum}
 
The closed expression for $\sg_{\e{eq}}^{(N)}$ in terms of the solution $\chi$ to the above Riemann--Hilbert problem and the close relation between the 
two-body interaction in the potential and the $\mc{S}_{N}$ operator's kernel allow one to exploit the system of jumps for $\chi$
so as to recast $\mc{E}_{N}^{(+)}\big[ \sg_{\e{eq}}^{(N)} \big]$ only in terms of $N$, $\ov{b}_N$ and $\chi$ evaluated at special points:
\begin{multline}
\mc{E}_{N}^{(+)}\big[ \sg_{\e{eq}}^{(N)} \big] \; = \;   \f{\varkappa }{ 2 N } \cosh(\ov{b}_N) \; + \;  \f{ \varkappa^2 \ex{2\ov{b}_N }  }{ 8 \pi  N^2 } 
 \Big\{  \chi_{12}^2(\ii)\, +  \, 2  \big[  \chi_{12}( \ii) \chi_{11}^{\prime}(\ii)  \, - \,  \chi_{11}( \ii) \chi_{12}^{\prime}(\ii)  \big]  \Big\}   \nonumber \\
 - \f{ \varkappa \ex{\ov{b}_N }  }{ 4   N } 
\Big\{ 1 \, + \,  \ex{-\ov{x}_N } \, + \, \chi_{22;-}(0) \big[ 2 \chi_{11}(\ii)+\ii \chi_{12}(\ii) \big] \, - \, 2 \chi_{21;-}(0)  \chi_{12}(\ii) \Big\}  \;. 
\end{multline}
Once that one arrives to the above closed expression, it is a matter of direct calculations which build on the uniform on $\Cx$ large-$N$ asymptotic expansion 
for $\chi$ provided by the non-linear steepest descent so as to infer the large-$N$ asymptotics

\begin{prop}
 
 One has the large-$N$ asymptotic behaviour 
\beq
\mc{E}_N^{(+)}\big[ \sg_{\e{eq}}^{(N)}\big] \; = \;  \f{ 3 \pi^{4} \, \mf{b} \,  \hat{\mf{b}}  \, \tilde{w}_1 }{ 4 ( \ov{b}_N)^3\, \tilde{w}_2   \, \mf{t}\big(2\ov{b}_N\big)  } \, + \, 
\f{ 9 \, \pi^4 \, \mf{b} \,  \hat{\mf{b}} }{ 8 ( \ov{b}_N)^4 \, \mf{t}^2\big( 2\ov{b}_N \big)  } 
\bigg\{    1  - \tfrac{ 2 \tilde{w}_1 }{ \ov{b}_N  \tilde{w}_2 }    \bigg\}  
\, + \,  
\e{O}\Big( \ex{-2\ov{b}_N(1-\eps )}\Big) \;, 
\enq
where $\mf{t}$ is as introduced in Theorem \ref{Theorem caracterisation minimiseur fnelle EN +} and we have rescaled the $w_k$ variables:
\beq
w_1\, =\, 2 \ov{b}_N \tilde{w}_1\; , \quad w_2 \, = \, 2 (\ov{b}_N)^2\,  \tilde{w}_2\; , \quad with \quad 
\tilde{w}_k \,=\, 1+ \e{O}\Big( \tfrac{ 1 }{ \ov{b}_N } \Big) \; \quad as \; \; N \rightarrow + \infty \;. 
\enq

\end{prop}

Together with Propositions \ref{Proposition estimation comportement gd N Z N p majorant}-\ref{Proposition minimiseur unique pour mathcal E de N et t} 
and the lower bound in \eqref{ecriture borne in sur minimum E n t avec minimum de EN plus}, the above theorem yields Theorem \ref{Theorem borne sup sur fct partition}.

\section{Conclusion}

In this paper we reviewed the bootstrap program approach to the rigorous construction of 1+1 dimensional integrable quantum field theories arising as appropriate quantisations
of integrable classical evolution equations of 1+1 dimensional field theory. This was done on the example of the Sinh-Gordon quantum field theory which is 
the simplest and non-trivial  instance of such model. The approach starts by proposing an appropriate Hilbert space on which such a model is realised. Then, it 
produces the form of the $\op{S}$-matrix which governs the scattering in such a case. This $\op{S}$-matrix arises as a solution of certain symmetry constrains 
on the scattering in a relativistically invariant theory along with the requirement of the factorisability of scattering into a concatenation of two-particle processes. 
Then, the quantum fields, which are operator valued distributions on functions 
of the space-time variables, are constructed as integral operators whose integral kernels satisfy a set of equations, the bootstrap program axioms $\mathrm{i)-v)}$, 
which should be taken as the basic axioms of the theory. These axiomatic equations strongly depend on the form of the $\op{S}$-matrix for the given theory. 
It turns out that the bootstrap program equations can be solved explicitly with the help of the algebraic setting provided by the quantum integrability of the model
and, in particular, the Yang-Baxter equation satisfied by the $\op{S}$-matrix. Once one ends up with the set of explicit solutions to $\mathrm{i)-v)}$, it remains to check the consistency of the
whole construction, in particular, that the so-constructed quantum fields do form an algebra and that they commute at space-like separations. 
The latter requirement is crucial for guaranteeing the causality of the so-constructed theory and thus it being viable as a \textit{per se} quantum field theory. 
To check these last steps of the construction, one must show that the series of multiple integrals resulting from the integral operator's multiplications do converge. 
This was a long standing open question in this field and its solution \cite{KozConvergenceFFSeriesSinhGordon2ptFcts}, in the simplest case scenario, was discussed by the author in the last section of this paper. 

There are still numerous open questions related to these topic. First of all, to implement the method of \cite{KozConvergenceFFSeriesSinhGordon2ptFcts} for establishing the 
convergence of form factor expansions for time-like separated two-point functions in the model  as well as for general multi-point correlation functions regardless of the space or time-like intervals between the operators. These questions do definitely seem to 
be manageable within a finite time. Further, to extend the methods of proving the convergence to more challenging but also more physically relevant  models such as the 
1+1 dimensional integrable Sine-Gordon quantum field theory. There, the multitude of asymptotic particles along with the presence of bound states and equal mass 
asymptotic particles will definitely be a challenging, but hopefully surmountable task. 

Last but not least, one should provide a thorough description of the correlation functions in the infrared limit, \textit{viz}. when the Minkowski separation
between the operator approaches zero. In the case of the two-point function given in \eqref{ecriture developpement serie fct 2 pts} that would correspond to extracting the $r \rightarrow 0^+$ limit.

\section*{Acknowledgment}

I thank M. Karowski, M. Lashkevich and F. Smirnov for stimulating discussions
on various aspects of integrable quantum field theories.

The author is supported by CNRS and also acknowledges support from  ERC Project LDRAM: ERC-2019-ADG Project 884584. This material is based upon work supported by the National Science Foundation
under Grant No. DMS-1928930 while the author participated in a program hosted 
by the Mathematical Sciences Research Institute in Berkeley, California, during 
the Fall 2021 semester.
 







\end{document}